\documentclass{article}
\usepackage{amsmath}
\usepackage{amsthm}
\usepackage{amssymb}
\usepackage{algorithm}
\usepackage{algorithmic}
\usepackage{xcolor}
\usepackage{physics}
\usepackage{esvect}
\usepackage{authblk}

\newtheorem{theorem}{Theorem}

\newtheorem{lemma}{Lemma}
\newtheorem{fact}{Fact}
\theoremstyle{definition}
\newtheorem{definition}{Definition}

\begin{document}
\title{Quantum and classical query complexities for determining connectedness of matroids}
\author{Xiaowei Huang, Shiguang Feng, Lvzhou Li\footnote{lilvzh@mail.sysu.edu.cn}}
\affil{Institute of Quantum Computing and and Software, School of Computer Science and Engineering, Sun Yat-sen University, Guangzhou 510006, China}

\maketitle


\begin{abstract}
Connectivity is a fundamental structural property of matroids, and has 
been studied algorithmically over 50 years. In 1974,  Cunningham proposed a deterministic algorithm consuming
$O(n^{2})$ queries to the independence oracle to determine whether a matroid is connected. Since then,  no algorithm, not even a random one, has worked better.
To the best of our knowledge, the  classical query complexity lower bound  and the quantum complexity 
for this problem have not been considered. Thus, in this paper we are devoted to addressing these issues, and our contributions are threefold as follows:
(i) First, we  prove that the randomized query complexity of determining whether a matroid is connected  is $\Omega(n^2)$ and thus the  algorithm proposed by Cunningham is  optimal in classical computing. (ii) Second, we present a quantum algorithm with $O(n^{3/2})$ queries, which exhibits provable quantum speedups over classical ones.
(iii) Third, we prove that any quantum algorithm requires  $\Omega(n)$ queries, which indicates that quantum algorithms can achieve at most a quadratic speedup over classical ones. Therefore, we have a relatively comprehensive understanding of the potential of quantum computing in determining  the connectedness of matroids.\\

\end{abstract}



\section{Introduction}

The concept of matroid was originally introduced by Whitney
\cite{JHUP/whitney35}
in 1935 as a generalization of the concepts of linear spaces and graphs. 
Now matroids have become ubiquitous in modern mathematics and computer 
science, and have various applications in geometry, topology, network theory, 
coding theory and quantum computing, especially in combinatorial  optimization \cite{DBLP:journals/tcad/AmyMM14,DBLP:journals/comsur/BassoliMRST13,
DBLP:journals/mp/Edmonds71,DBLP:conf/ismp/Iri82,HRW/lawler76,DBLP:mann2021simulating,
shepherd2009temporally}.
In combinatorial optimization, many important problems can be described and solved based on  
 matroid structures, which include greedy 
algorithms \cite{edmonds1971matroids}, matroid intersection problem \cite{cunningham1986improved,edmonds1979matroid,lawler1975matroid}, 
matroid partition problem \cite{cunningham1986improved,edmonds1968matroid}, 
etc. In addition, many matroid problems 
are interesting and valuable in their own right from an algorithmic viewpoint.

\subsection{Connectivity of Matroid}

Connectivity is a fundamental structural property of matroids. The notion of 
connectivity in matroid theory  derived from the connectivity of 
graphs \cite{tutte1966connectivity4graphs}. 
Since then, much work has been devoted to matroid connectivity, and 
it has become one of the most interesting 
directions in matroid theory
\cite{bixby1979matroids,cunningham1981on,oxley1981on,oxley1982on,seymour1988on,tutte1966connectivity}. 

Let $M$ be a matroid on the ground set $E$ with rank function $r$. A partition
$\{E_1,E_2\}$ of $E$ is an  \emph{$m$-separation} of $M$ ($m\geq 1$), if
$ \min\{|E_1|,|E_2|\}\geq m$ and $r(E_1)+r(E_2)-r(E)\leq m-1.$
The \emph{connectivity} of $M$,  denoted by 
$\lambda(M)$,  is the 
minimum $k$ such that $M$ has a $k$-separation. If $M$ has no such $k$, then $\lambda(M)=\infty$.
In general, a matroid $M$ is said to be \emph{$k$-connected} for $k\geq 2$, if 
$\lambda(M)\geq k$. A 2-connected matroid is usually called a
\emph{connected}  or \emph{non-separable} matroid (see Definition \ref{de:non-separable-matroid}).



\subsection{Motivation and Problem Statement}

Quantum computers have been shown to  outperform classical computers in solving some problems with significant application value \cite{DBLP:conf/stoc/Grover96,DBLP:conf/focs/Shor94}.  Discovering  more problems
for  which quantum computers  exhibit speedups over classical ones is one of the primary goals in the field of  quantum
computing. 
Considering matroids are abstractions of some mathematical structures common to graph 
and linear space,  it is of great interest  and meaningful to explore what matroid 
problems can be faster solved on quantum computers than classical ones. 
To the best of our knowledge, there are only few discussions on potential quantum 
speedups in solving matroid problems
\cite{huang2021quantum,huang2022quantum,DBLP:conf/ciac/KulkarniS13}. 
Therefore, in this paper we are trying to advance research in this direction.
We investigate both quantum and classical query complexities for determining whether
a matroid is connected, which is formally described as follows:\\

\noindent\emph{\textbf{Connected Matroid Problem (CMP).} Given a matroid $M$ accessed 
by the independence oracle, determine whether $M$ is connected.}\\

Determining whether a graph is biconnected\footnote{An undirected graph is called biconnected if there are two vertex-disjoint paths between any two  vertices, that is the graph is still connected after deleting any one vertex.} is a fundamental problem in graph theory. Biconnected graphs are instances of connected matroids. A graph is biconnected if and only if the corresponding cycle matroid is connected. It is of importance to study the Connected Matroid Problem.

\subsection{Our Contributions}
 Given a matroid $M=(E,\mathcal{I})$ with $|E|=n$ accessed by the
independence oracle, we prove the following main results.
\begin{enumerate}
\item Any randomized algorithm that determines whether $M$ is connected with
bounded error 1/3 needs at least $\Omega(n^2)$ queries to the independence oracle
(see Theorem \ref{th:radomized-query-complexity}).
\item There is a quantum algorithm that determines whether $M$ is connected 
with bounded error 1/3 using $O(n^{3/2})$ queries to the independence oracle
(see Theorem \ref{th:quantum-algorithm-for-non-separable-testing}).
\item Any quantum algorithm that determines whether $M$ is connected with
bounded error 1/3 needs at least $\Omega(n)$ queries to the independence oracle
(see Theorem \ref{th:quantum-query-complexity}).
\end{enumerate}

Based on these results, we have a relatively comprehensive understanding of the potential of quantum computing in determining the connectedness of matroids.  For the problem of determining whether a matroid is connected, one can see  that (i) Cunningham's deterministic algorithm is 
 optimal in classical computing, (ii) quantum computers exhibit provable  quantum speedups over classical ones, and  (iii) quantum computers
can achieve at most a quadratic speedup over classical ones.

To prove the lower bounds, 
we construct a disconnected matroid from a connected matroid with the least number of bases (also known as the minimal matroid 
\cite{dinolt1970an}) by removing a base 
(see Lemma~\ref{le:nearest-to-minimal-matroid}). Intuitively, the resulting  
disconnected matroid and the minimal matroid are relatively hard to distinguish. 
This is the keypoint for the proof of both the quantum and classical lower bounds. 

\subsection{Related Work} 
In 1974, Cunningham~\cite{cunnigham1974a} first considered  \textbf{CMP} and proposed an algorithm 
using $O(n^2)$ queries to the independence oracle, where $n$ is the cardinality of the 
ground set. In 1977, Krogdahl~\cite{krogdahl1977dependence} proposed an algorithm similar to that of Cunningham with the same query complexity when studying the dependence graphs of bases in matroids.
Bixby and Cunningham~\cite{bixby1979matroids} gave an algorithm for determining 
3-connected. For  $k>3$, Cunningham \cite{cunnigham1974a} presented an algorithm 
to determine whether a matroid is $k$-connected. 
Jensen and Korte 
\cite{DBLP:journals/siamcomp/JensenK82} showed that there is 
no polynomial query algorithm to compute the connectivity of a matroid.

Since Cunningham proposed the $O(n^2)$ algorithm for \textbf{CMP} in 1974, no algorithm, not even a random one, has worked better.
As far as we know, the  classical query complexity lower bound and  the quantum complexity
for this problem have not been considered. It is not known whether 
Cunningham's algorithm is optimal and whether  there are more efficient quantum 
algorithms to determine the connectedness of matroids.


\subsection{Organization } 
The remainder of this paper is organized as follows. 
In Section 2, some basic concepts in matroid theory and query model are introduced.
In Section 3, we prove the lower bound on randomized query complexity.
In Section 4, we give a quantum algorithm determining the connectedness of matroids.
In Section 5, we prove the lower bound on quantum query complexity.
In Section 6, we conclude this paper.

\section{Preliminaries}\label{pre}
\noindent\textbf{Notations}: 
Let $E$ be a finite set, $A\subseteq E$, and $x\in E$. We use $|E|$, $A+x$, and $A-x$ to denote the cardinality of $E$, the set $A\cup\{x\}$, and the set $A-\{x\}$, respectively. For a positive integer $n$, we use $[n]$ to denote the set $\{1,\dots,n\}$. For a string $x\in\{0,1\}^n, i\in[n]$, $x_i$ denotes
the $i$-th bit of $x$.

\subsection{Matroid Theory}\label{sec:matroid}
Here we give some basic definitions and concepts on matroids. 
Matroid theory is established as a generalization of linear algebra and graph
theory. Some concepts are similar to those of linear algebra or graphs.
We refer the reader to \cite{OUP/oxley11} and \cite{welsh1976matroid} for more details about matroid theory.

\begin{definition}[\textbf{Matroid}]
  \label{def:matroid}
  A \emph{matroid} is a combinational object defined by the tuple
  $M=(E,\mathcal{I})$ on the finite ground set $E$ and $\mathcal{I}\subseteq 2^E$ such
  that the following properties hold:
  \begin{enumerate}
  \item[\textbf{I0}.] $\emptyset \in \mathcal{I}$;
  \item[\textbf{I1}.] If $A'\subseteq A$ and $A\in\mathcal{I}$, then $A'\in\mathcal{I}$;
  \item[\textbf{I2}.] For any two sets $A,B\in\mathcal{I}$ with $|A|<|B|$, there 
  exists an element $x\in B-A$ such that $A+x\in\mathcal{I}$.
  \end{enumerate}
\end{definition}
The members of $\mathcal{I}$ are the \emph{independent} sets of $M$.
The subsets of $E$ not belonging to $\mathcal{I}$ are called \emph{dependent}.
A \emph{base} of $M$ is a maximal independent set, and the collection of bases is denoted by $\mathcal{B}(M)$.
A \emph{circuit} of $M$ is a minimal dependent set, and the collection of circuits is denoted by $\mathcal{C}(M)$.

\begin{definition}[\textbf{Rank}]
  \label{def:rank}
  The \emph{rank function} of a matroid $M=(E,\mathcal{I})$ is the function
  $r:2^E\rightarrow\mathbb{Z}$ defined by
  \begin{equation}
   r(A) = \max\Big\{|X|:X\subseteq A,X\in\mathcal{I}\Big\}\;\;\;\;\;\;(A\subseteq E).
   \end{equation}
\end{definition}
\noindent{}The rank of $M$, denoted by $r(M)$, is $r(E)$.

\begin{lemma}[Base Axioms]\label{th:base-axioms}
A non-empty collection $\mathcal{B}$ of subsets of $E$ is the set of 
bases of a matroid on $E$ if and only if it satisfies the following 
condition:
\begin{itemize}
\item[(B1)] If $B_1,B_2\in\mathcal{B}$ and $x\in B_{1}-B_{2}$, then there exists $y\in B_{2}-B_{1}$ such that $B_{1}+y-x\in\mathcal{B}$.\label{B1}
\end{itemize}
\end{lemma}

\begin{lemma}[Circuit Axioms]\label{th:circuit-axioms}
A non-empty collection $\mathcal{C}$ of subsets of $E$ is the set of 
circuits of a matroid on $E$ if and only if it satisfies the following 
conditions:
\begin{itemize}
\item[(C1)] If $X,Y\in\mathcal{C}$ and $X\neq Y$, then $X\nsubseteq Y$. \label{C1}
\item[(C2)] If $C_{1},C_{2}$ are distinct members of $\mathcal{C}$ and
$z\in C_{1}\cap C_{2}$, then there exists $C_{3}\in\mathcal{C}$ such that
$C_{3}\subseteq (C_{1}\cup C_{2})-z$.\label{C2}
\end{itemize}
\end{lemma}

The proofs of the two lemmas above can be found in \cite{OUP/oxley11,welsh1976matroid}.

\begin{fact}
Let $M=(E,\mathcal{I})$ be a matroid. If $B$ is a base of $M$ and $x\in E- B$, then
there exists a unique circuit $C(x,B)$ such that
\begin{equation*}
    {C(x,B)}\subseteq{B}+x, \text{ and } x\in{C(x,B)}.
\end{equation*}
\end{fact}
\noindent{}The circuit $C(x,B)$ is called the fundamental circuit of $x$ in the base $B$.

\begin{definition}[\textbf{Connected Matroid}]
\label{de:non-separable-matroid}
A matroid $M = (E,\mathcal{I})$ with rank function $r$ is called \emph{connected} if every
nonempty proper subset $A$ of $E$ satisfies
\begin{equation*}
r(A)+r(E-A)>r(E). 
\end{equation*}
Otherwise $M$ is \emph{disconnected}.
\end{definition}


\begin{lemma}\label{le:connected-matroid}
A matroid $M = (E,\mathcal{I})$ is \emph{connected} if and only if for every pair of distinct elements $x$ and $y$ of $E$ 
there is a circuit of $M$ containing $x$ and $y$.
\end{lemma}

The proof of Lemma \ref{le:connected-matroid} can be found in \cite{OUP/oxley11,welsh1976matroid}.

\begin{definition}[\textbf{Independence Oracle}]
  \label{def:matroid-oracle}
  Let $M=(E,\mathcal{I})$ be a matroid. An \emph{independence oracle} of $M$ is a function $O_{i}:2^{E}\rightarrow\{0,1\}$, for $S\in E$, such that 
\begin{equation*}
 O_{i}(S) = \left\{
 \begin{array}{ll}
  1  & \text{if}\; S\in \mathcal{I}, \\
  0  & \text{otherwise}.
 \end{array}
 \right.
\end{equation*}

  In the quantum setting, we also use $O_i$ as the quantum independence 
  oracle that maps $|S\rangle|b\rangle$ to
  $|S\rangle|b\oplus O_{i}(S)\rangle$.
\end{definition}

\subsection{Query Models and Quantum Computing}
In this subsection, we briefly introduce the query complexity for deterministic, randomized, and quantum computing models. For more details about the query complexity and quantum computing, we refer the reader to \cite{DBLP:journals/tcs/BuhrmanW02,CUP/nielsen10}.

\subsubsection*{Decision Tree}
Let $f:\{0, 1\}^{n}\rightarrow\{0,1\}$ be a Boolean function. A \emph{(deterministic) decision tree} 
(which is also called a deterministic query algorithm) for $f$ is a 
binary tree $T$ in which each internal node is labelled with a variable $x_i$, and has two outgoing edges labelled with 0 and 1, respectively. Each leaf node is labelled with 0 or 1. Given an input $x$, the tree $T$ is evaluated as follows. Start at the root node, say with label $x_j$, get the value of $x_j$ by querying $x$, if $x_j=1$ then choose the subtree reached by taking the 1-edge, otherwise, choose the subtree reached by taking the 0-edge. Repeat the procedure recursively until a leaf node is reached. The label of the leaf node is $f(x)$.


With the decision tree, we can define the complexity measures for $f$.
\begin{definition}
The cost of a deterministic decision tree $T$ on an input $x$, denoted by $cost(T, x)$, 
is the number of queries made on $x$ when evaluating $T$.
\end{definition}

\begin{definition}
The deterministic query complexity of a function $f$, denoted by $D(f)$, is
\begin{equation}
D(f) = \min_{T\in\mathcal{T}}\max_{x\in\{0,1\}^n} cost(T,x),
\end{equation}
where $\mathcal{T}$ is the set of decision trees that compute $f$.\\
\end{definition}

\subsubsection*{Randomized Decision Tree}

There are two equivalent ways to define a randomized decision tree.
The first way is to define a randomized decision tree in which the branch taken
at each node is determined by the query value or a random coin flip
(possibly biased). The second is to define a randomized decision tree
as a probability distribution $P$ over deterministic decision trees. 
The randomized decision tree is evaluated by choosing a deterministic decision tree according to
$P$.

We say that a randomized decision tree computes $f$ with bounded-error if its 
output equals $f(x)$ with probability at least 2/3, for all $x\in\{0,1\}^n$.

\begin{definition}
The cost of a randomized decision tree over a probability distribution $P$ on an
input $x$, denoted by $cost(P, x)$, is the expected number of queries made on $x$ over $P$,
\begin{equation}
 cost(P,x)=\sum_{T\in\mathcal{T}}P(T)cost(T,x). 
\end{equation}
\end{definition}

\begin{definition}
The randomized query complexity of a function $f$, denoted by $R(f)$, is
\begin{equation}
R(f) = \min_{P\in\mathcal{P}}\max_{x\in\{0,1\}^n} cost(P,x),
\end{equation}
where $\mathcal{P}$ is the set of probability distributions over deterministic decision trees that compute $f$ with bounded-error.
\end{definition}

Let $A$ be a deterministic algorithm that outputs 1 or 0 for any input $x$ of the domain of $f$, let $\epsilon(A,x)=0$ if $A$ outputs $f(x)$, and $\epsilon(A,x)=1$ otherwise. 
\begin{definition}
Let $\mu$ be a probability distribution over the domain of $f$. We say that a deterministic algorithm $A$ computes $f$ over $\mu$ if
\begin{equation*}
\sum_{x\sim\mu}\mu(x)\epsilon(A,x) < 1/3.
\end{equation*}
\end{definition}


\begin{definition}
Let $\mu$ be a probability distribution over the domain of $f$. The distributional query complexity of $f$ on $\mu$, denoted by $D_{\mu}(f)$, is 
\begin{equation*}  D_{\mu}(f)=\min_{A\in\mathcal{A}_{\mu}}\sum_{x\sim\mu}\mu(x)cost(A,x),
\end{equation*}
where $\mathcal{A}_{\mu}$ is the set of deterministic algorithms that compute $f$
over $\mu$.
\end{definition}

\begin{lemma}[Yao’s minimax principle~\cite{DBLP:conf/focs/Yao77}]
Let $\mathcal{D}\subseteq \{0,1\}^n$, for any Boolean function 
$f:\mathcal{D}\rightarrow\{0,1\}$,
\begin{itemize}
\item[(1)] $D_{\mu}(f)\leq R(f)$ for every distribution $\mu$ over $\mathcal{D}$, and
\item[(2)] $D_{\mu}(f)=R(f)$ for some distribution $\mu$ over $\mathcal{D}$.
\end{itemize}
\end{lemma}

In general, it is  difficult to find the randomized query 
complexity of a problem directly. Yao's minimax principle gives a 
possible way to find  the randomized query complexity, which is one 
of the most commonly used methods.\\

\subsubsection*{Quantum Decision Tree}

Let $x$ be a binary string of length $n$. In the quantum query models~\cite{DBLP:conf/focs/BealsBCMW98}, we use $O_{x}$ to denote the query oracle for $x$, and define $O_{x}$ as a unitary operation
\begin{equation}
O_{x}:|i,y\rangle\rightarrow|i,y\oplus{x_i}\rangle,
\end{equation}
where $y$ is one qubit, $i$ is an index of $\lceil \log n\rceil$ qubits for every bit of $x$,  the variable $x_i$ denotes the $i$-th bit of $x$, and $\oplus$ is the \emph{exclusive-or} operator.

A \emph{quantum decision tree} on an input $x$ runs as follows: start with an initial state $\ket*{\vv{0}}$, then apply a unitary operation $U_{0}$ to the state, then  apply a query $O_{x}$, then apply another 
unitary operation $U_{1}$, and so on. Thus, a $T$-query quantum decision tree corresponds to a sequence $U=U_TO_x\cdots O_xU_1O_xU_0$ of unitary operations $\{U_i:i=0,\dots,T\}$ and $O_x$. Here each $U_i$ $(0\leq i \leq T)$ is a unitary operation that is independent of $x$. 
The result is obtained by measuring the final state $U\ket*{\vv{0}}$.

We say that a $T$-query quantum decision tree computes $f$ with bounded-error if the computing result equals $f(x)$ with probability at least  2/3. 
The quantum query complexity of a function $f$, denoted by $Q(f)$, is the least $T$ such that there is a $T$-query quantum decision tree that computes $f$ (with bounded-error 1/3).
Note that the quantum query complexity only considers the number of queries to the oracle $O_x$, but not $U_i$.\\

\begin{fact}
$Q(f)\leq R(f)\leq D(f)$.
\end{fact}

\section{Lower Bound on Randomized Query Complexity}
In this section, we prove the randomized query complexity of \textbf{CMP}. We first introduce a 
special class of connected matroids of rank $r$ on $n$ elements 
with $r(n-r)+1$ bases that are studied by Dinolt~\cite{dinolt1970an}. This kind of matroid is called \emph{minimal matroid}. Then we prove a key result,  Lemma~\ref{le:nearest-to-minimal-matroid}, which shows that a 
disconnected matroid can be obtained by removing any base of a minimal
matroid. Intuitively, the resulting disconnected matroid and the minimal matroid 
are relatively hard to distinguish. This is the keypoint for the proof of both 
quantum and classical lower bounds.

\begin{lemma}[\cite{dinolt1970an}]\label{le:dinolt-matroid}
A connected matroid with rank $r$ on a ground set with cardinality $n$
has at least $r(n-r)+1$ bases. For any positive integers $n$ and $r$ with $n>r$,
there exists one (up to isomorphism) connected matroid $M=(E,\mathcal{I})$
with $|E|=n,r(E)=r$ and $|\mathcal{B}(M)|=r(n-r)+1$.
\end{lemma}

\begin{lemma}[\cite{dinolt1970an}]\label{le:minimal-matroid}
Let $r$ and $n$ be two integers with $0< r<n$, and $E=\{e_1,\dots,e_r,e_{r+1},\dots,e_{n}\}$, $E_0=\{e_1,\dots,e_r\}$,
$\overline{E_{0}}=E-E_{0}=\{e_{r+1},\dots,e_{n}\}$, 
$\mathcal{C}=\{E_{0}+e_{j}:e_{j}\in\overline{E_{0}}\}$$\cup$
$\{\{e_{j},e_{j'}\}:e_{j},e_{j'}\in\overline{E_0},j\neq j'\}$.
Then $\mathcal{C}$ is the set of circuits of a minimal matroid on $E$ with rank $r$.
\end{lemma}
\begin{proof}
By the Circuit Axioms, it is easy to verify that $\mathcal{C}$ is the set of circuits
of a matroid on $E$ with rank $r$. We denote the matroid by $M(\mathcal{C})$. 
For any two distinct elements $x$ and $y$ of $E$, it is not hard to find an element
of $\mathcal{C}$ that contains both $x$ and $y$. This implies that $M(\mathcal{C})$
is connected.

Furthermore, we can get the set $\mathcal{B}(M(\mathcal{C}))$ of bases of $M(\mathcal{C})$ is 
\[\mathcal{B}(M(\mathcal{C})) =\{E_{0}\}\cup \{E_{0}-e_{i}+e_{j}:e_{i}\in E_{0},e_{j}\in\overline{E_{0}}\},\]
and $|\mathcal{B}(M(\mathcal{C}))|=r(n-r)+1$. Therefore, $M(\mathcal{C})$ is a minimal matroid.
\end{proof}

In the following, we prove a fundamental lemma that is used to prove the lower bounds 
of both the quantum and classical query complexity. Generally, the removal (a base) operation is not a 
well-defined operation on matroids. But for some special matroids (for instance, the minimal 
matroids), we can construct new matroids by the removal operations.

\begin{lemma}\label{le:nearest-to-minimal-matroid}
Let $M$ be the minimal matroid defined by $\mathcal{C}$ on $E$ as in the proof of Lemma~\ref{le:minimal-matroid}, $\mathcal{B}$ the set of bases of $M$. 
Then for any $B\in\mathcal{B}$, there is a disconnected matroid $M_{B}$ on $E$ with rank $r$ such that $\mathcal{B}-B$ is the set of bases of $M_{B}$.
\end{lemma}
\begin{proof}
We prove that the set $\mathcal{B}-B$ satisfies the condition in the Base Axioms. For $r<n$,
it is easy to see that $\mathcal{B}-B$ is not empty. To prove $\mathcal{B}-B$ satisfies
the condition (B1), we proceed by considering the following two cases.

\textbf{Case\;1}. If $B=E_0$, we show that $\mathcal{B}_{1}=$
$\mathcal{B}-B=\{E_{0}-e_{i}+e_{j}:e_{i}\in E_{0},e_{j}\in\overline{E_{0}}\}$
satisfies the Base Axioms. Let $B_{1}=E_{0}-x_{1}+y_{1}$ and 
$B_{2}=E_{0}-x_{2}+y_{2}$, where $x_{1},x_{2}\in E_{0}$ and 
$y_{1},y_{2}\in\overline{E_{0}}$. Assume that $B_{1}\neq B_{2}$.

\begin{description}
\item[(1)] If $x_{1}=x_{2}$, then it implies $y_{1}\neq y_{2}$. We have 
$B_{1}+y_{2}-y_{1}=B_{2}\in\mathcal{B}_{1}$.

\item[(2)] If $y_{1}=y_{2}$, then it implies $x_{1}\neq x_{2}$. We have 
$B_{1}+x_{1}-x_{2}=B_{2}\in\mathcal{B}_{1}$.

\item[(3)] If $x_{1}\neq x_{2}$ and $y_{1}\neq y_{2}$, then
$B_{1}-B_{2}=\{x_2,y_{1}\},B_{2}-B_{1}=\{x_{1},y_{2}\}$.
We have $B_{1}+x_{1}-x_{2}=E_{0}-x_{2}+y_{1}\in\mathcal{B}_{1}$.
\end{description}
Thus, we can see that $\mathcal{B}_{1}$ satisfies the Base Axioms.

\textbf{Case\;2}. If $B=E_{0}-x+y$ for $x\in E_{0},y\in\overline{E_{0}}$, we show that $\mathcal{B}_{2}=\mathcal{B}-B$ satisfies the Base Axioms. 
It is easy to check that $B_{1}+y-x=B_{2}\in\mathcal{B}_{2}$ if $B_{1}=E_{0}, B_{2}=E_{0}-x'+y'\neq B$ for $x'\in E_{0}$ and $y'\in\overline{E_{0}}$.

Let $B_{1}=E_{0}-x_{1}+y_{1}\neq B$ and $B_{2}=E_{0}-x_{2}+y_{2}\neq B$,
where $x_{1},x_{2}\in E_{0}$ and $y_{1},y_{2}\in\overline{E_{0}}$. Assume that 
$B_{1}\neq B_{2}$. The cases $x_{1}=x_{2}$ and $y_{1}=y_{2}$ are the same as \textbf{(1)} and \textbf{(2)} in \textbf{Case\;1}, respectively. If $x_{1}\neq x_{2}$ and $y_{1}\neq y_{2}$, then $B_{1}-B_{2}=\{x_{2},y_{1}\}$ and $B_{2}-B_{1}=\{x_{1},y_{2}\}$. We have  $B_{1}+x_{1}-y_{1}=E_{0}\in\mathbf{B}_2$. Thus, we can see that $\mathcal{B}_{2}$ satisfies the Base Axioms.

Combining \textbf{Case\;1} and \textbf{Case\;2}, we know that for any $B\in\mathcal{B}$, $\mathcal{B}-B$ is the set of bases of a matroid on $E$ with rank $r$. We denote this 
matroid as $M_{B}$. Because $r(M_{B})=r$ and $|\mathcal{B}-B|=r(n-r)$, by Lemma~\ref{le:dinolt-matroid}, it follows that $M_{B}$ is disconnected.
\end{proof}

The following theorem is the main result of this section.
\begin{theorem}\label{th:radomized-query-complexity}
Given a matroid $M$ on a ground set with cardinality $n$ accessed by 
the independence oracle, any randomized algorithm 
that determines  whether $M$ 
is connected with bounded error $1/3$ needs at least $\Omega(n^2)$ 
queries to the independence oracle.
\end{theorem}

\begin{proof}
 We shall show that the lower bound of the distributional 
 complexity of \textbf{CMP} is $\Omega(n^2)$. Then by Yao’s 
 minimax principle, we have $R(\textbf{CMP})=\Omega(n^2)$.

Let $M$ be a minimal matroid with rank $r$ on a ground set
of cardinality $n$, $N=r(n-r)+1$, $\mathcal{B}=\{B_i:i\in[N]\}$ the set
of bases of $M$. For any $B\in\mathcal{B}$, let $M_{B}$ be the matroid whose set of bases is determined by the set $\mathcal{B}-B$ as defined in Lemma~\ref{le:nearest-to-minimal-matroid}. 

We begin by picking the input distribution $x\sim\mu$ defined as follows:
\begin{itemize}
\item[(1)] Sample $i\in[N]$ uniformly at random.
\item[(2)] Set $x=M$ or $x=M_{B_i}$ with probability $\frac{1}{2}$.
\end{itemize}
Fix any deterministic algorithm $A$ computing
\textbf{CMP} on the distribution $\mu$ with bounded error
 $1/3$. We use $T$ to denote the number of queries
made by $A$. Let $K=\{i_1,\dots,i_T\}$ be the set of queries made
by $A$ where the answer to every one of these queries is 1 (this means 
that what we are querying is a base of the input matroid).
Since $A$ is deterministic, the set $K$ is fixed and well-defined.
For any input $x\sim\mu$ to the algorithm $A$, we have
\begin{equation}
 \Pr(i\in K)=\frac{|K|}{N}=\frac{T}{N}.    
\end{equation}
Define $E(x)$ as the event that the algorithm $A$ succeeds on computing $x$.
At the same time, conditioned on $i\notin K$, the input is $M$ or $M_{B_i}$
with probability $\frac{1}{2}$. Hence,
\begin{equation}
 \Pr(E(x)|i\notin K)=\frac{1}{2}.    
\end{equation}
Combining the above, we have
\begin{equation}
 \begin{array}{rcl}
 2/3&\leq&\Pr(E(x))_{x\sim\mu}\\
 &=&\Pr(i\in K)\Pr(E(x)|i\in K)+
 \Pr(i\notin K)\Pr(E(x)|i\notin K)\\
 &\leq&\frac{T}{N}+(1-\frac{T}{N})\cdot\frac{1}{2},
 \end{array}   
\end{equation}
which implies $T\geq N/3$. When $r=n/2$, we get $T>n^{2}/12$. Consequently, 
$D_{\mu}(\textbf{CMP})=\Omega(n^2)$, thus the randomized
query complexity of \textbf{CMP} is $\Omega(n^2)$.
\end{proof}

\section{Quantum Algorithm for Determining Connected Matroids}
In this section, we give a quantum algorithm to determine whether a matroid is connected.
Firstly, we introduce the partial representation of matroids and the 
corresponding bipartite graphs. The relation between the connectedness of 
matroids and the connectedness of the corresponding bipartite graphs is 
given in Lemma~\ref{le:relationship-matroid-and-bipartite-graph}. This leads 
to an efficient algorithm to determine the connectedness of matroids. 
Secondly, we construct the oracle for the adjacency matrix model of the
corresponding bipartite graph from the independence oracle of the matroid.
Finally, we introduce the quantum search algorithm in Lemma~\ref{le:grover-algorithm}. Combining the above, we give a quantum algorithm. \\

\paragraph{Partial Representation.} Partial representation is introduced by
Truemper~\cite{truemper1984partial}. Let $M$ be a matroid on a finite ground set 
$E$, and $B$ a base of $M$. The \emph{partial representation} of $M$ with respect
to $B$ is a $\{0, 1\}$-matrix $P$ with rows indexed on elements $x\in B$ and columns
indexed on elements $y\in\overline{B}=E-B$ such that the entry $P(x,y)$ is 1 if and
only if $x\in C(y,B)$, otherwise $P(x,y)$ is 0. A bipartite graph $G(P)$ derived from
$P$ is obtained by setting one node for each row and column of $P$ and one edge for each
nonzero entry. Krogdahl~\cite{krogdahl1977dependence} proved the relationship of 
the connectedness between $M$ and $G(P)$ in Lemma~\ref{le:relationship-matroid-and-bipartite-graph}.

\begin{lemma}[\cite{krogdahl1977dependence}]
\label{le:relationship-matroid-and-bipartite-graph}
Let $M$ be a matroid, $B$ a base of $M$, $P$ the partial representation corresponding
to $B$ and $G(P)$ the associated bipartite graph. Then $M$ is connected if and only
if $G(P)$ is connected.
\end{lemma}


Based on an idea similar to Lemma~\ref{le:relationship-matroid-and-bipartite-graph}, 
Cunningham~\cite{cunnigham1974a} gave a very pretty and efficient algorithm 
determining whether a matroid is connected, which takes $O(|E|^2)$ queries to the 
independence oracle. \\


\paragraph{Adjacency Matrix Model.} In the adjacency matrix model, an undirected 
graph $G=(V,A)$ is given as an adjacency matrix 
$P\in\{0,1\}^{|V|\times |V|}$ such that $P_{ij}=1$ if and only if $(v_i,v_j)\in A$ for $v_i,v_j\in V$.

Let $M$ be a matroid on $E$, $B$ a base of $M$, $P$ 
the partial representation corresponding to $B$ and $G(P)$ the 
associated bipartite graph. Since $G(P)$ is a bipartite graph,
$P$ is essentially the adjacency matrix of $G(P)$. Therefore, 
it is not difficult to construct an oracle $O$ that accesses $P$
from the independence oracle $O_{i}$ as follows:
\begin{align}\label{def:oracle_from_matroid}
     O(j,k)=\left\{
      \begin{array}{ll}
        O_{i}(B+k-j)&\text{if}\;\; j\in B\;\;\text{and}\;\; k\in \overline{B},\\
        0&\text{otherwise}.
      \end{array}
      \right.
\end{align}

In~\cite{durr2006quantum}, D\"{u}rr et al. give a quantum algorithm to determine 
the connectedness of graphs and show that the quantum query complexity is 
$\Theta(n^{3/2})$ in the adjacency matrix model, where $n$ is the number of 
vertices in the graph. Since we only need to determine the connectedness of $G(P)$
without finding its spanning tree, and deal with bipartite graphs, we give 
a quantum algorithm that is actually more efficient with the quantum depth-first search method.\\

\paragraph{Quantum search.} A search problem on a set $[n]$ with $n$ elements is a subset $J\subseteq [n]$ with the
characteristic function $f:[n]\rightarrow\{0,1\}$ such that
\begin{equation*}
 f(x) = \left\{\begin{array}{ll}
  1  & \text{if}\;x\in J, \\
  0  & \text{otherwise}.
 \end{array}
 \right.
\end{equation*}
Any $x\in J$ is called a solution of the search problem. 

\begin{lemma}[Grover's algorithm \cite{WOL/Boyer98,DBLP:conf/stoc/Grover96}]
  \label{le:grover-algorithm}
  Let $J$ be a search problem on a set $[n]$ with $n$ elements, $|J|=k$ and 
  $f$ the characteristic function of $J$. Given a search space 
  $S\subseteq [n]$ with $|S|=N$, the expected quantum query 
  complexity for finding one solution of $J$ is $O(\sqrt{N/k})$ for $k>0$.
  Whether $J\cap S$ is empty can be determined in $O(\sqrt{N})$
  quantum queries to $f$ with success probability at least 2/3. 
\end{lemma}

Our main result in this section is as follows.
\begin{theorem}\label{th:quantum-algorithm-for-non-separable-testing}
Given a matroid $M$ on a ground set with cardinality $n$ accessed by 
the independence oracle, there is a quantum 
algorithm that determines  whether $M$ is connected with bounded error 
$1/3$ using $O(n^{3/2})$ queries to the independence oracle.
\end{theorem}

\begin{proof}
Consider Algorithm~\ref{al:test-matroid-connected-quantum}, 
after finding a base $B$ of $M$ using a greedy algorithm, 
it uses the depth-first search method to determine whether the associated
bipartite graph is connected. In Line 5, the oracle $O$ used by Grover's
algorithm is defined by Equation~\eqref{def:oracle_from_matroid}.
Thus, it is not difficult to see the correctness of Algorithm~\ref{al:test-matroid-connected-quantum}.

\begin{algorithm}
  \caption{A quantum algorithm deciding whether a matroid is connected.}
  \label{al:test-matroid-connected-quantum}
  \begin{algorithmic}[1]
    \REQUIRE A matroid $M$ accessed through the independence oracle $O_i$.
    \ENSURE YES if $M$ is connected, otherwise NO.
    \STATE $B\leftarrow\;$Find a base of $M$ with greedy algorithm.
    \STATE Let $s\in B$, PUSH($S$,$s$).\textcolor{green!50!black}{//S is a stack, and PUSH, TOP, and POP are operations of the stack.}
    \WHILE{$S$ is not empty}
    \STATE $u\leftarrow$ TOP($S$). 
    \STATE Use Grover's algorithm to find an unmarked vertex $v$ with $O(u,v)=1$.
    \IF{such a vertex $v$ exists}
    \STATE PUSH($S$,$v$)
    \ELSE 
    \STATE $t\leftarrow$ POP($S$), mark $t$ as visited.
    \ENDIF
    \ENDWHILE
    \IF {all the vertices are visited}
    \RETURN YES.
    \ELSE
    \RETURN NO.
    \ENDIF
  \end{algorithmic}
\end{algorithm}

Next, we prove the query complexity of Algorithm~\ref{al:test-matroid-connected-quantum}.
Let $M$ be defined on $E$ whose cardinality is $n$, $|B|=r$, and $\overline{B}=E-B$. Finding a base $B$ 
of $M$ needs $O(n)$ queries to the independence oracle. For any $s\in B$, every vertex is 
discovered by the depth-first search at most once. In the adjacency matrix model,
if $u\in B$, then the next vertex is found by $O(\sqrt{n-r})$ quantum queries to the
independence oracle. If $u\in\overline{B}$, it takes $O(\sqrt{r})$ queries.
The total query is $O(n)+(n-r)O(\sqrt{n-r})+rO(\sqrt{r})$, which is bounded by $O(n^{3/2})$.
\end{proof}

\section{Lower Bound on Quantum Query Complexity}
In this section, we prove the quantum query complexity of \textbf{CMP}. Firstly, we introduce the 
quantum adversary method proposed by Ambainis~\cite{DBLP:conf/stacs/AmbainisS06} in Lemma~\ref{le:ambainis}, which is 
one of the common methods to prove the lower bound of quantum query 
complexity. In order to use the quantum adversary method to prove the 
quantum query complexity of \textbf{CMP}, we need to encode matroids to 
0-1 strings as done in graphs, where each bit represents a subset of the 
ground set. Then we can express \textbf{CMP} by a Boolean function.
And the function is a partial function since not every string is a 
representation of a matroid.
Combining Lemma~\ref{le:minimal-matroid} and 
Lemma~\ref{le:nearest-to-minimal-matroid}, we can give a lower bound for 
quantum complexity of \textbf{CMP}.

\begin{lemma}[Ambainis’s quantum adversary method \cite{DBLP:conf/stacs/AmbainisS06}]
\label{le:ambainis}
 Let $f$ be a function of $n$ $\{0,1\}$-valued variables, and $X,Y$ two sets of inputs 
 such that $f(x)\neq f(y)$ if $x\in X$ and $y\in Y$. Let $R\subseteq X\times Y$ be a
 relation such that
\begin{itemize}
\item [1.] for every $x\in X$, there are at least $m$ different $y\in Y$ such that
$(x,y)\in R$;
\item [2.] for every $y\in Y$, there are at least $m'$ different $x\in X$ such that
$(x,y)\in R$;
\item [3.] for every $x\in X$ and $i\in[n]$, there are at most $l$ different $y\in Y$ 
such that $(x,y)\in R$ and $x_i\neq y_i$;
\item [4.] for every $y\in Y$ and $i\in[n]$, there are at most $l'$ different $x\in X$ 
such that $(x,y)\in R$ and $x_i\neq y_i$.
\end{itemize}
Then any quantum algorithm computing $f$ with bounded error 1/3 takes at least 
$\Omega(\sqrt{\frac{mm'}{ll'}})$ queries.
\end{lemma}

Our main result in this section is as follows.
\begin{theorem}\label{th:quantum-query-complexity}
Given a matroid $M$ on a ground set with cardinality $n$ accessed by 
the independence oracle, any quantum algorithm 
that determines  whether $M$ is connected with bounded error $1/3$ 
needs at least $\Omega(n)$ queries to the independence oracle.
\end{theorem}

\begin{proof}
Let $E$ be a finite ground set with cardinality $n$, $\mathcal{M}$ the set 
of matroids on $E$. We define a mapping 
\begin{equation}
 \chi:\mathcal{M}\rightarrow\{0,1\}^{2^n}, 
\end{equation}
for every matroid $M\in\mathcal{M}$. Let $\chi(M)$ be the unique string 
encoding $M$, where different bits represent
different subsets of $E$ and the appearance of `1' indicates that the corresponding subset is an independent set of $M$.

Define a Boolean function
\begin{equation}
 f:\mathcal{D}\rightarrow\{0,1\}   
\end{equation}
that computes the connectedness of matroids on $E$,  where $\mathcal{D}\subseteq\{0,1\}^{2^n}$ 
contains the valid strings encoding the matroids on $E$. For any $x\in\mathcal{D}$, $f(x)=1$
indicates that the matroid encoded by $x$ is connected, otherwise $f(x)=0$.
Let $M_{0}$ be a minimal matroid of rank $n/2$ on $E$, and $\mathcal{B}$ the 
set of bases of $M_{0}$. Let 
$X=\{\chi(M_{0}\}$, $Y=\{\chi(M_{B}):B\in\mathcal{B}\}$, where $M_{B}$ is a disconnected
matroid defined in Lemma~\ref{le:nearest-to-minimal-matroid}. By Lemma~\ref{le:ambainis}, we can see that any quantum algorithm computing $f$ with bounded 
error 1/3 takes at least $\Omega(\sqrt{|X||Y|})=\Omega(n)$ queries.
\end{proof}

\section{Conclusion}
In this paper, we consider the query complexity and quantum algorithm for determining
the connectedness of matroids. We show that the randomized query complexity is 
$\Omega(n^2)$ and thus the deterministic algorithm proposed by Cunningham is 
optimal. We also give a quantum algorithm taking $O(n^{3/2})$ queries 
and prove that the quantum query complexity is $\Omega(n)$, which indicates that
quantum algorithms can achieve at most a quadratic speedup over classical ones.
Therefore, we have a relatively comprehensive understanding of the potential of quantum
computing in determining the connectedness of matroids.
Obviously, there is still a gap between the lower and upper bounds of the quantum query complexity.
It is quite possible to believe that this quantum lower bound is not the optimal one. 
Finding a tight lower bound for the quantum query complexity of this problem can be the future work of our research.

\bibliographystyle{acm}
\bibliography{reference}
\end{document}